\documentclass[conference,letterpaper%,onecolumn
]{IEEEtran}

\usepackage[utf8]{inputenc} 
\usepackage[T1]{fontenc}
\usepackage{url}
\usepackage{ifthen}
\usepackage{cite}
\usepackage[cmex10]{amsmath} % Use the [cmex10] option to ensure complicance
\usepackage[table]{xcolor}
\usepackage{times}
\usepackage{epsfig}
\usepackage{amsmath}
\usepackage{amsthm}
\usepackage{amsfonts}
\usepackage{graphicx}
\usepackage{amssymb}
\usepackage{amstext}
\usepackage{latexsym}
\usepackage{color,colortbl}
\usepackage{ifthen}
\usepackage{multirow}
\usepackage{verbatim}
\usepackage{array,tabularx}
\usepackage{arydshln}
\usepackage[mathscr]{euscript}
\usepackage{accents}
\usepackage{cite}
\usepackage{hhline}
\usepackage{caption}
\usepackage{subcaption}
\usepackage{enumerate}
\usepackage{xcolor}
\usepackage{mathtools}
\usepackage{url}
\usepackage{xparse}
\usepackage{makecell}
\usepackage{varwidth}
\usepackage{bm}
\usepackage{arydshln}
\usepackage{comment}
\usepackage{wrapfig}

\usepackage{caption}
\usepackage{float}
\usepackage{booktabs}

\usepackage{nidanfloat}
\usepackage{subcaption}

\usepackage{mathtools}

\usepackage{booktabs}

\usepackage{float}

\usepackage{multirow}

\usepackage{gnuplottex}

\usepackage{filecontents}

\newcolumntype{C}[1]{>{\centering\let\newline\\\arraybackslash\hspace{0pt}}m{#1}}

\usepackage[normalem]{ulem}

\usepackage{amsmath,pgfplots,amssymb}
\usepackage{subcaption,graphicx}

\newtheorem{theorem}{Theorem}

\newtheorem{proposition}{Proposition}

\theoremstyle{definition}

\newtheorem{corollary}{Corollary}

\theoremstyle{definition}

\theoremstyle{definition}

\newcommand{\F}{\mathbb{F}}

\usetikzlibrary{matrix,decorations.pathreplacing}

\begin{document}
\title{Notes on Communication and Computation in Secure Distributed Matrix Multiplication}

% \author{
% Rafael G.L. D'Oliveira, Salim El Rouayheb, Daniel Heinlein, David Karpuk\\
% ECE, Rutgers University, USA\\
% Department of Communications and Networking, Aalto University, Finland\\
% Departamento de Matem\'aticas, Universidad de los Andes, Colombia\\
% Emails: \{rafael.doliveira, salim.elrouayheb\}@rutgers.edu, daniel.heinlein@aalto.fi, da.karpuk@uniandes.edu.co
% }

\author{
   \IEEEauthorblockN{Rafael G. L. D'Oliveira\IEEEauthorrefmark{1},
                     Salim El Rouayheb\IEEEauthorrefmark{2},
                     Daniel Heinlein\IEEEauthorrefmark{3},
                     and David Karpuk\IEEEauthorrefmark{4}}
  \IEEEauthorblockA{\IEEEauthorrefmark{1}
                    RLE,  Massachusetts Institute of Technology, USA,
                    rafaeld@mit.edu}
  \IEEEauthorblockA{\IEEEauthorrefmark{2}
                     ECE, Rutgers University, USA,
                     salim.elrouayheb@rutgers.edu}
   \IEEEauthorblockA{\IEEEauthorrefmark{3}%
                     Department of Communications and Networking, Aalto University, Finland,
                     daniel.heinlein@aalto.fi}
   \IEEEauthorblockA{\IEEEauthorrefmark{4}%
                    F-Secure, Finland,
                    da.karpuk@uniandes.edu.co}
 }

\maketitle

%%%%%%
%% Abstract: 
%% If your paper is eligible for the student paper award, please add
%% the comment "THIS PAPER IS ELIGIBLE FOR THE STUDENT PAPER
%% AWARD." as a first line in the abstract. 
%% For the final version of the accepted paper, please do not forget
%% to remove this comment!
%%
\begin{abstract}
We consider the problem of secure distributed matrix multiplication in which a user wishes to compute the product of two matrices with the assistance of honest but curious servers. 

\noindent In this paper, we answer the following question: Is it beneficial to offload the computations if security is a concern? We answer this question in the affirmative by showing that by adjusting the parameters in a polynomial code we can obtain a trade-off between the user's and the servers' computational time.

\noindent Indeed, we show that if the computational time complexity of an operation in $\mathbb{F}_q$ is at most $\mathcal{Z}_q$ and the computational time complexity of multiplying two $n\times n$ matrices is $\mathcal{O}(n^\omega \mathcal{Z}_q)$ then, by optimizing the trade-off, the user together with the servers can compute the multiplication in $\mathcal{O}(n^{4-\frac{6}{\omega+1}} \mathcal{Z}_q)$ time.

\noindent We also show that if the user is only concerned in optimizing the download rate, a common assumption in the literature, then the problem can be converted into a simple private information retrieval problem by means of a scheme we call Private Oracle Querying. However, this comes at large upload and computational costs for both the user and the servers.
\end{abstract}

\let\thefootnote\relax\footnote{\scriptsize{This work was supported in part by NSF Grant CCF 1817635 and CNS 1801630.}}

%% The paper must be self-contained. However, if you are referring to
%% a full version for checking certain proofs, please provide the
%% publically accessible location below.  If the paper is completely
%% self-contained, you can remove the following line from your9004505
%% submission.

\section{Introduction}

There has been a growing interest in applying coding theoretic methods for  Secure Distributed Matrix Multiplication (SDMM) \cite{ravi2018mmult,Kakar2019OnTC,d2019gasp,DOliveira2019DegreeTF, Aliasgari2019DistributedAP,Kakar2019UplinkDownlinkTI,Yu2020EntangledPC}. In SDMM, a user has two matrices, ${A \in \mathbb{F}_q^{r \times s}}$ and $B \in \mathbb{F}_q^{s \times t}$, and is interested in obtaining $AB \in \mathbb{F}_q^{r \times t}$ with the help of $N$ servers without leaking any information about $A$ or $B$ to any server.  All servers are assumed to be honest and responsive, but are curious, in that any $T$ of them may collude to try to deduce information about either $A$ or $B$. The original performance metric used in the literature is the download cost \cite{ravi2018mmult}, i.e. the total amount of data downloaded by the user from the servers, with later work considering the total communication cost\cite{Chang2019OnTU,9004505,Kakar2019UplinkDownlinkTI}.

In \cite{jia2019capacity}, the following existential issue is raised with the SDMM setting: Is it beneficial to offload the computations if security is a concern? Indeed computing the product $AB$ locally is both secure and has zero communication cost. The authors in \cite{jia2019capacity} circumvent this by changing the setting so that the user does not possess the matrices $A$ and $B$. This forces communication to be the only way for the user to obtain the product $AB$. This however, is solving a problem quite different from the one initially posed in \cite{ravi2018mmult}.

In this paper we revisit the original setting of SDMM and show that offloading the computations can be justified from a computational perspective. More precisely, we show that by adjusting the parameters in a polynomial code we can obtain a trade-off between the user's and the servers' computational time, as shown in Figure~\ref{fig:Fa}. Indeed, if the computational time complexity of an operation in $\mathbb{F}_q$ is at most $\mathcal{Z}_q$ and the computational time complexity of multiplying two $n\times n$ matrices is $\mathcal{O}(n^\omega \mathcal{Z}_q)$ then, by optimizing the trade-off, the user together with the servers can compute the multiplication in $\mathcal{O}(n^{4-\frac{6}{\omega+1}} \mathcal{Z}_q)$ time, as shown in Figure~\ref{fig:Fb}.

\begin{figure*}[!t]
  \begin{subfigure}[b]{0.5\linewidth}
  \centering
\includegraphics[width=9cm]{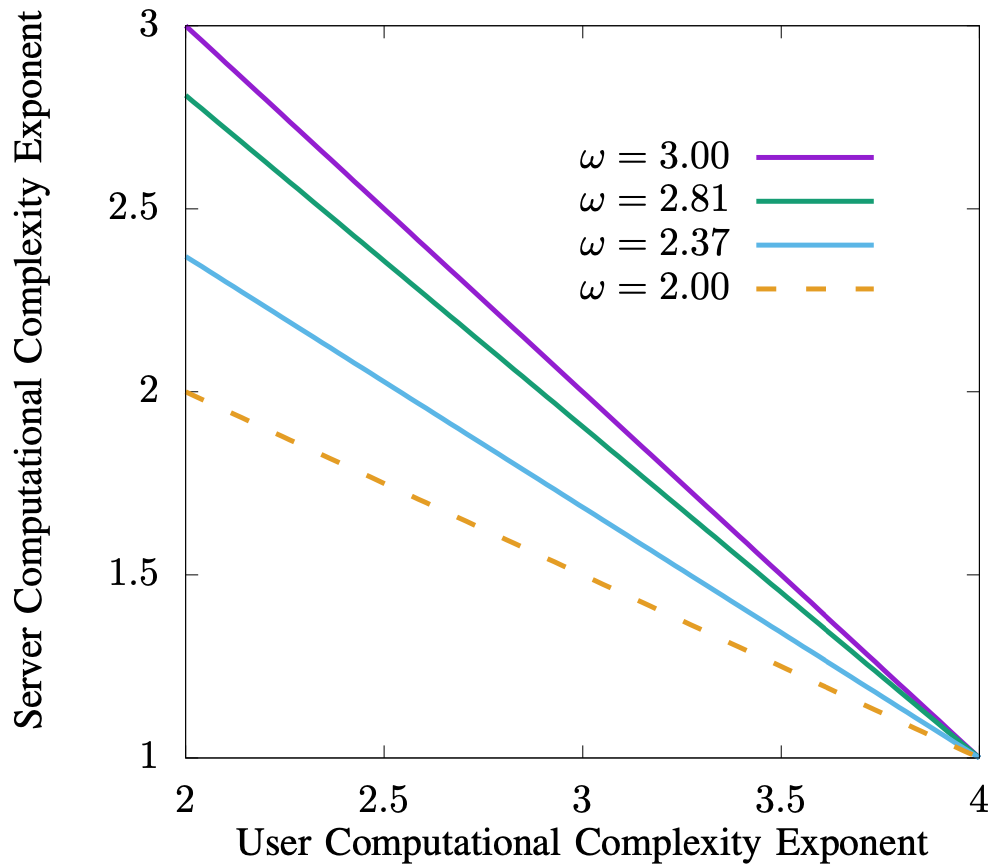}
\caption{Trade-off between the user and servers' computation exponents, e.g. if the servers use the standard matrix multiplication algorithm, $\omega =3$, and the user sets its computational complexity exponent to $2.5$ then the server will also have a computational complexity exponent of $2.5$.}
\label{fig:Fa}
  \end{subfigure} \hspace{5pt}
  \begin{subfigure}[b]{0.5\linewidth}
  \centering
\includegraphics[width=9cm]{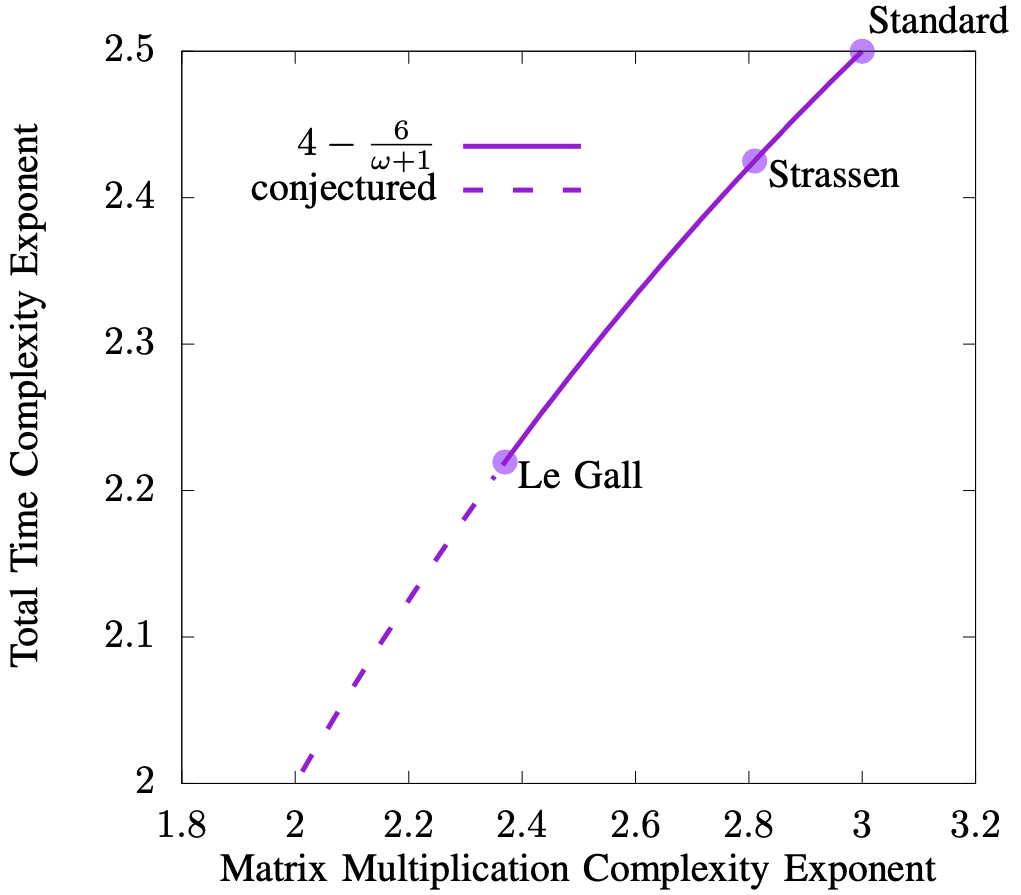}
\caption{Total time complexity for GASP when choosing best trade-off between user and servers' time complexity exponent as a function of the exponent of the matrix multiplication algorithm, $\omega$. We show that $\varepsilon$ can be chosen so that the total time complexity has exponent $4-\frac{6}{\omega+1}$.}
\label{fig:Fb}
  \end{subfigure} 
  \caption{The figures pertain to the setting in Section \ref{sec:param} where we analyze the computational complexity of GASP codes. In this setting, $r=s=t=n$, the security parameter $T$ is a constant, the partitioning parameters $K=L=n^\varepsilon$, and the servers use a matrix multiplication algorithm with computational complexity $\mathcal{O}(n^\omega)$.}
  \label{fig:F} 
\end{figure*}

\subsection{Related Work}

For distributed computations, Polynomial codes were originally introduced in \cite{polycodes1} in a slightly different setting, namely to mitigate stragglers in distributed matrix multiplication. This work was followed by a series of works \cite{polycodes2,pulkit,pulkit2,fundamental}. However, the polynomial codes in these works are not designed to ensure security, making them not applicable to settings where there are privacy concerns related to the data being used.

\subsection{Main Contributions}

The main contributions of this work are as follows.
\begin{itemize}

    \item In Section \ref{sec: private oracle}, we show that if the performance metric for SDMM is solely the download cost, then, by transforming the problem into a private information retrieval problem, we can obtain download costs much lower than those obtained using polynomial codes. This, however, comes at exponential upload and computational costs. The scheme, however, can be readily implemented in settings where the download cost is the performance metric of interest, like in \cite{Chang2019OnTU} or \cite{jia2019capacity}.
    
    \item In Section \ref{sec:param}, we show the existence of a regime under which outsourcing computations with security constraints is beneficial. We do this by analyzing the computational time complexity of a family of polynomial codes known as gap additive secure polynomial (GASP) codes \cite{d2019gasp,DOliveira2019DegreeTF}, and show that by adjusting the code parameters we can obtain a trade-off between the user's and the servers' computational time. By optimizing this trade-off we can show that if the time complexity of an operation in $\mathbb{F}_q$ is at most $\mathcal{Z}_q$ and a matrix multiplication algorithm for $n \times n$ matrices with time complexity $\mathcal{O}(n^\omega \mathcal{Z}_q)$ is used, then the total time taken for the user to retrieve $AB$ with the help of the servers is given by $\mathcal{O}(n^{4-\frac{6}{\omega+1}} \mathcal{Z}_q)$.

\end{itemize}

\section{Notation}

Our analysis in sections \ref{sec: polynomial codes} and \ref{sec:param} will require asymptotic notation for multivariate functions. As shown in \cite{howell2008asymptotic}, care must be taken when generalizing the asymptotic notation from univariate to multivariate functions.

Hence, we apply the following asymptotic notation.
For a function $f$ mapping $D \subseteq \mathbb{R}^n$ to $\mathbb{R}$, such that $D$ is in each coordinate not upper bounded, $\mathcal{O}(f(x))$ is the set of all $g : D \rightarrow \mathbb{R}$ such that there exist $N,c \in \mathbb{R}_+$ with $|g(x)| \le c |f(x)|$ for all $x$ with $N \le x_i$ for all $i \in \{1,\ldots,n\}$.
We define $\Omega(f(x))$ in the same way with the inequality replaced by $|g(x)| \ge c |f(x)|$.

% We use the following asymptotic notation.

% \begin{itemize}
%     \item $f(n) = \mathcal{O}(g(n))$ if and only if $\limsup\limits_{n\rightarrow \infty} \frac{|f(n)|}{g(n)} < \infty$.
%     \item $f(n) = \Omega (g(n))$ if and only if $\liminf\limits_{n\rightarrow \infty} \frac{|f(n)|}{g(n)} > 0$.
%     \item $f(n) = \Theta (g(n))$ if and only if $f(n) = \mathcal{O}(g(n))$ and $f(n) = \Omega (g(n))$.
% \end{itemize}

We assume a base field $\mathbb{F}_{p}$ over which all elementary operations (addition, subtraction, multiplication, division) take constant time. We also assume that transmitting symbols in $\mathbb{F}_{p}$ between the user and the servers takes constant time. 

When constructing polynomial codes we will need to consider a field extension $\mathbb{F}_q$ of $\mathbb{F}_{p}$. We assume that any elementary operation or generation of a random element in $\mathbb{F}_q$ takes time at most $\mathcal{Z}_q$. The possible values for $\mathcal{Z}_q$ depend on the representation of the field elements, e.g. powers of a generator of the group of units $\mathbb{F}_q^\times$ or polynomials in $\mathbb{F}_p[X] / (f)$ (with $f \in \mathbb{F}_p[X]$ irreducible and of degree $d$ with $p^d=q$), and of the underlying machine, e.g. a Turing machine or a Boolean circuit~\cite{gashkov2013complexity}, and its implementation~\cite{katz1996handbook,brent2010modern}.

We set $\mathcal{Z}_q = \mathcal{O}({\log(q)}^\gamma)$, i.e. $\mathcal{Z}_q$ is polylogarithmic. If only additions and multiplications are used, for example, we can set $\gamma = 2$ if we use standard polynomial multiplication. This can be reduced by using better multiplication algorithms.

% We assume that any addition, subtraction, multiplication, and generation of a random element in $\mathbb{F}_q$, as well as each operations of a fast matrix multiplication algorithm of square matrices over $\mathbb{F}_q$ have computational cost at most $\mathcal{Z}_q$, which also depends on the representation of field elements, e.g. powers of a generator of the group of units $\mathbb{F}_q^\times$ or polynomials in $\mathbb{F}_p[X] / (f)$ (with $f \in \mathbb{F}_p[X]$ irreducible and of degree $d$ with $p^d=q$), and of the underlying machine, e.g. a Turing machine or a boolean circuit~\cite{gashkov2013complexity}, and its implementation~\cite{katz1996handbook,brent2010modern}.

Next, we assume that the transmission of one $q$-ary symbol has communication cost at most $\mathcal{C}_q$. If we use the usual polynomial representation, then $\mathcal{C}_q = \mathcal{O}(\log(q))$.

We denote by $\mathcal{M}(r,s,t)$ the computation complexity of multiplying an $r\times s$ matrix by an $s\times t$ matrix. The study of the computational complexity of matrix multiplication is one of the main topics in algebraic complexity theory.

The most understood case is for square matrices, i.e. when $r=s=t=n$. In \cite{strassen1969gaussian}, Strassen presented the first algorithm outperforming the standard $\mathcal{O}(n^3)$. Strassen's algorithm has computational complexity $\mathcal{O}(n^{\log_2(7)}) \approx \mathcal{O}(n^{2.81})$. This was further improved to $\approx \mathcal{O}(n^{2.37})$ by Coppersmith, Winograd, and Le Gall~\cite{le2014powers,coppersmith1990matrix}. Since any entry of both $n \times n$ matrices has to be used in general, the number of operations is at least $\Omega(n^2)$. It is an open problem if there exists an algorithm which uses $\Theta(n^2)$ operations.

\section{Private Oracle Querying} \label{sec: private oracle}

In this section we show that by transforming the SDMM problem into a private information retrieval problem we can obtain schemes with download costs much lower than polynomial codes. These schemes, however, have exponential upload and computational costs. They serve as an example of why we cannot use the download cost as the sole performance metric as was done originally in the literature.

The scheme, however, can be readily implemented in settings where the download cost is the performance metric of interest, like in  \cite{Chang2019OnTU} or \cite{jia2019capacity}.

We name this scheme a private oracle querying scheme and begin by giving a simplified example of it. It consists in transforming the secure distributed matrix multiplication problem into a private information retrieval problem \cite{chor1995private}.

The reason for naming it Oracle Querying, is that the technique applies to settings more general than matrix multiplication. Indeed the same can be done even for non-computable functions, say if the servers have access to some oracle.

\subsection{An Example}

Let $A,B \in \mathbb{F}_2$ and the number of servers be $N=2$ none of which collude, thus $T=1$, $r=s=t=1$, and $q=2$. The user is interested in $AB \in \mathbb{F}_2$. The Private Oracle Querying scheme consists in transforming SDMM into a private information retrieval problem.

The servers begin by precomputing all $M = q^{s(r+t)} = 4$ possible multiplications, shown in Table~\ref{tab:example}. Then, each server stores all possible multiplications in its database, i.e. the third column of Table~\ref{tab:example}.

\begin{table}[h] 
    \centering
    \begin{tabular}{c c c}
       $A$  & $B$ & $AB$ \\
       \toprule
       $0$ & $0$ & $0$ \\
       $0$ & $1$ & $0$ \\
       $1$ & $0$ & $0$ \\
       $1$ & $1$ & $1$ \\
    \end{tabular}
    \caption{Each server stores the third column in the table.}
    \label{tab:example}
\end{table}   

The user can obtain the multiplication privately via a private information retrieval scheme where the user wants one file out of the database, $\bm{D}$, of $M = q^{s(r+t)} = 4$ files each one of length $rt=1$. This can be done, for example, by using a simple secret sharing scheme achieving a download rate of $\mathcal{D} = \frac{N-T}{N}= \frac{1}{2}$, as shown in Table \ref{tab:pir example}.

\begin{table}[h] 
    \centering
    \begin{tabular}{c c c}
       & Server $1$  & Server $2$ \\
       \toprule
       Query: & $\bm{q}$ & $\bm{q}+\bm{e}_i$ \\
       Response: & $\langle \bm{D} , \bm{q} \rangle$ & $\langle \bm{D} , \bm{q}+\bm{e}_i \rangle$
    \end{tabular}
    \caption{To privately retrieve the $i$-th entry in the database from the servers the user generates a vector $\bm{q} \in \mathbb{F}_2^{Mrt}$ uniformly at random. He then sends $\bm{q}$ to Server $1$ and $\bm{q}+\bm{e}_i$ to Server $2$ where $\bm{e}_i$ is the $i$-th vector in the standard basis of $\mathbb{F}_2^{Mrt}$. The Servers perform an inner product of the received query with their database and sends it back to the user. The user then retrieves the $i$-th entry in the database by subtracting the responses.}
    \label{tab:pir example}
\end{table}

\subsection{The Scheme}

We now present the scheme, which we refer to as private oracle querying.

\begin{theorem} \label{teo:equivalence}
Let $N$ be the number of servers, $T$ the security parameter, $A \in \mathbb{F}_q^{r\times s}$ and $B \in \mathbb{F}_q^{s\times t}$. Then, the secure distributed matrix multiplication problem for computing ${AB \in \mathbb{F}_q^{r\times t}}$ can be solved by solving a private information retrieval problem where each server has $M = q^{s(r+t)}$ files, each one of length $rt$.
\end{theorem}

\begin{proof}
As a preprocessing step of the scheme, each server computes all $M = q^{s(r+t)}$ possible matrix multiplications and stores them in its database. Considering each result of each multiplication as a file, each server then has $M$ files, each of size $rt$. Thus, the secure distributed matrix multiplication problem can be reinterpreted as a private information retrieval problem where each server has $M$ files, each of size $rt$.
\end{proof}

If the field $q$ is large enough, the user can use a simple secret sharing scheme.

\begin{corollary} \label{cor: secret share}
Under the same hypothesis of Theorem \ref{teo:equivalence}, for large enough field size $q$, there exists a secure distributed matrix multiplication scheme with download rate $\mathcal{D} = \frac{N-T}{N}$.
\end{corollary}

\begin{proof}
This rate can be achieved by using the construction in Section III B of \cite{8906051}. The large field size is needed to guarantee the existence of an MDS code.
\end{proof}

The download capacity for private information retrieval is known \cite{sun2017capacity}. However, as the number of files grows, this capacity converges to the rate in Corollary \ref{cor: secret share}.

If one uses the download rate as the sole performance metric for the setting in \cite{ravi2018mmult}, these private information retrieval codes can outperform the polynomial codes in \cite{ravi2018mmult,Kakar2019OnTC,d2019gasp,DOliveira2019DegreeTF, Aliasgari2019DistributedAP,Kakar2019UplinkDownlinkTI,Yu2020EntangledPC}. They, however, have two shortcomings.

First, the upload cost is exponential, since even a single query will have the size of the whole database, $q^{s(r+t)}$.

Second, the time to generate a single query, $\Omega(q^{s(r+t)})$, is much longer than the time for the user to calculate the matrix multiplication locally using the standard matrix multiplication algorithm, $\mathcal{O}(rst \mathcal{Z}_q)$.

In other settings, where the user does not have access to both matrices and computational costs are not considered, like in \cite{Chang2019OnTU} or \cite{jia2019capacity}, private oracle querying can be readily applied.

% \subsection{Time Complexity}

% \begin{table}[h]
%     \centering
%     \begin{tabular}{ll}
%     \toprule
%     Operation &   Time Complexity  \\
%     \midrule
%     Query Encoding &  $\mathcal{O}((NT+T-T^2) q^{s(r+t)} rt\mathcal{Z}_q)$ \\
%     User Decoding &   $\mathcal{O}((T+1)rt\mathcal{Z}_q)$ \\
%     Each Server &     $\mathcal{O}(q^{s(r+t)} rt\mathcal{Z}_q)$\\
%     \bottomrule
%     \end{tabular} 
%      \caption{Computation Complexity}
%         \label{tab:poq computational}
% \end{table}

% \begin{enumerate}
% \item 
% \item 
% \item We assume that each server has precomputed the database of $M = q^{s(r+t)}$ matrices $AB$ for all $A$ and all $B$. Then, the work per server amounts to a linear combination of received coefficients and this database. The number of additions in $\mathbb{F}_q$ is $(M-1)rt$ and the number of multiplications in $\mathbb{F}_q$ is $Mrt$.
% \end{enumerate}

% \begin{table}[h]
%     \centering
%     \begin{tabular}{ll}
%     \toprule
%     Operation &   Time Complexity  \\
%     \midrule
%     Upload &          $\mathcal{O}(N q^{s(r+t)} \mathcal{C}_q)$ \red{$\mathcal{O}(N q^{s(r+t)} rt\mathcal{C}_q)$} \\
%     Download &        $\mathcal{O}(\frac{N-T}{N} rt\mathcal{C}_q)$ \\
%     \bottomrule
%     \end{tabular}
%         \caption{Communication Complexity}
%             \label{tab:poq communication}
% \end{table}

% \begin{enumerate}
%     \item Each server gets a $q$-ary vector of length $M$ for the inner product with its database.
%     \item 
% \end{enumerate}

\section{Polynomial Codes} \label{sec: polynomial codes}

Polynomial codes for secure distributed matrix multiplication were first introduced in \cite{ravi2018mmult} and later improved on in  \cite{ravi2018mmult,Kakar2019OnTC,d2019gasp,DOliveira2019DegreeTF, Aliasgari2019DistributedAP,Kakar2019UplinkDownlinkTI,Yu2020EntangledPC}. Our goal is to highlight the existence of a regime where securely offloading the computation to the workers is beneficial. Towards that goal we analyze the communication and computation complexity of a family of polynomial codes called GASP codes \cite{d2019gasp,DOliveira2019DegreeTF}. Since we are using upper bounds to show that SDMM is beneficial, constructions which outperform GASP codes will also be upper bounded by the expressions shown. The analysis shown here can be straightforwardly extended to other polynomial codes in the literature.

\subsection{Constructing GASP Codes}

Let $A \in \mathbb{F}_q^{r \times s}$ and $B \in \mathbb{F}_q^{s \times t}$ be partitioned as follows:
\begin{equation*} \label{eq:outer product}
    A = \begin{bmatrix}
    A_1 \\ \vdots \\ A_K
    \end{bmatrix},\quad 
    B = \begin{bmatrix} 
    B_1 & \cdots & B_L
    \end{bmatrix},
\end{equation*}

\begin{equation*}\label{partition1}
    \text{so that}\quad 
    AB = \begin{bmatrix}
    A_1B_1 & \cdots & A_1B_L \\
    \vdots & \ddots & \vdots \\
    A_KB_1 & \cdots & A_KB_L
    \end{bmatrix}.
\end{equation*}

The user chooses $T$ matrices $R_t$ over $\F_q$ of the same size as the $A_k$ independently and uniformly at random, and $T$ matrices $S_t$ of the same size as the $B_\ell$ independently and uniformly at random. A polynomial code is a choice of ${\alpha = (\alpha_1,\ldots,\alpha_{K+T})\in\mathbb{N}^{K+T}}$ and
${\beta = (\beta_1,\ldots,\beta_{L+T})\in\mathbb{N}^{L+T}}$ defining the polynomials
\[
f(x) = \sum_{k = 1}^KA_kx^{\alpha_k} + \sum_{t = 1}^TR_tx^{\alpha_{K+t}},
\]

\[
g(x) = \sum_{\ell = 1}^LB_\ell x^{\beta_\ell} + \sum_{t = 1}^TS_t x^{\beta_{L+t}}
\]
and their product $h(x) = f(x)g(x)$.

Given $N$ servers, the user chooses evaluation points $a_1,\ldots,a_N\in \F_{q^r}$ for some finite extension $\F_{q^r}$ of $\F_q$.  They then send $f(a_n)$ and $g(a_n)$ to server $n = 1,\ldots,N$, who computes the product $f(a_n)g(a_n) = h(a_n)$ and transmits it back to the user.  The user then interpolates the polynomial $h(x)$ given all of the evaluations $h(a_n)$, and attempts to recover all products $A_kB_\ell$ from the coefficients of $h(x)$.

GASP codes \cite{d2019gasp,DOliveira2019DegreeTF} are a family of polynomial codes constructed via a combinatorial table called the degree table.

In Table \ref{tab:poly communication} we show the upload and download time complexity for GASP codes. These values follow directly from the analysis done in Appendix B of \cite{9004505}.

\begin{table}[h]
    \centering
    \begin{tabular}{ll}
    \toprule
    Operation & Time Complexity  \\
    \midrule
    Upload & $\mathcal{O}(N s(\frac{r}{K}+\frac{t}{L})\mathcal{C}_q)$\\
    Download & $\mathcal{O}(N \frac{rt}{KL}\mathcal{C}_q)$\\
    \bottomrule
    \end{tabular}
     \caption{Communication time for GASP codes.}
        \label{tab:poly communication}
\end{table}

\subsection{The Computational Complexity of GASP codes}

In this section we perform an analysis on the computational time complexity of GASP codes. The computations can be separated into three parts.
\begin{enumerate}
    \item \textbf{User Encoding}: the computation time it takes the user to generate the evaluations that will be uploaded to the servers.
    \item \textbf{Server Computation}: the computation time it will take each server to multiply the two evaluations it receives from the user.
    \item \textbf{User Decoding}: the computation time it will take the user to decode the matrix multiplication from what it received from the servers.
\end{enumerate}

\begin{table}[h]
    \centering
    \begin{tabular}{ll}
    \toprule
    Operation & Time Complexity  \\
    \midrule
    User Encoding &  $\mathcal{O}(Ns(r+t+T(\frac{r}{K} + \frac{t}{L}))\mathcal{Z}_q)$\\
    Server Computation & $\mathcal{O}(\mathcal{M}\left(\frac{r}{K},s,\frac{t}{L} \right)\mathcal{Z}_q)$ \\
    User Decoding & $\mathcal{O}(Nrt\mathcal{Z}_q)$ \\
    \bottomrule
    \end{tabular}
     \caption{Computation time for GASP codes.}
        \label{tab:poly computational}
\end{table}

\begin{theorem} \label{teo: comp time gasp}
The computational time complexity for GASP codes is given in Table \ref{tab:poly computational}.
\end{theorem}

\begin{proof}
\noindent

\begin{enumerate}
    \item \textbf{User Encoding}: The number of additions and multiplications in $\mathbb{F}_q$ needed to compute an evaluation of $f$ and $g$ are $(K+T)\frac{rs}{K}$ and $(L+T) \frac{st}{L}$. The result follows from performing this $N$ times, once for each server.
    \item \textbf{Server Computation}: Each server must compute the product of two matrices of dimensions $\frac{r}{K} \times s$ and $s \times \frac{t}{L}$.
    \item \textbf{User Decoding}: We assume that the inverted generalized Vandermonde matrix is precomputed. Then, the interpolation of $A_iB_j$ is a linear combination of the servers' answers. The number of additions in $\mathbb{F}_q$ is $KL(N-1)\frac{r}{K}\frac{t}{L}$ and the number of multiplications is $KLN\frac{r}{K}\frac{t}{L}$.
\end{enumerate}
\end{proof}

If using the standard matrix multiplication algorithm then we can substitute $\mathcal{M}\left(\frac{r}{K},s,\frac{t}{L} \right) = \mathcal{O}\left(\frac{rst}{KL}\right)$.

\section{Choosing the Right Parameters}\label{sec:param}

In this section we show that, by choosing the right parameters for GASP codes, secure distributed matrix multiplication can speed up the computation time when compared to the user performing the computation locally.

We will analyze the following setting. We consider square matrices, i.e. $r=s=t=n$, assume that the security parameter, $T$, is a constant, and that the partitioning parameter $K=L=n^\varepsilon$ for some $0 \le \varepsilon \le 1$. We also assume that the servers multiply two $n \times n$ matrices using an algorithm with computational complexity $\mathcal{O}(n^\omega)$. Our goal is to study the time complexity of GASP codes as $n$ grows.

In \cite{DOliveira2019DegreeTF}, it was shown that for GASP codes we have the bounds $KL \le N \le (K+T)(L+T)$. Thus, $N = \mathcal{O}(K^2)$.

We calculate the time complexity for each of the servers.

\begin{proposition} \label{pro: server comp time}
Let $r=s=t=n$, $T$ be a constant, $K=L$, and $\mathcal{O}(n^\omega)$ be the computational complexity of the algorithm which the servers use to multiply an $n \times n$ matrix. Then, the time complexity for each server to compute the matrix multiplication sent to it in the GASP scheme is $\mathcal{O}(\tfrac{n^\omega}{K^{\omega-1}} \mathcal{Z}_q)$.
\end{proposition}

\begin{proof}
The rectangular matrices each server has to multiply, say $F$ and $G$, have dimensions $\frac{n}{K} \times n$ and $n \times \frac{n}{K}$, so that they can be split into
\begin{align*}
F = \begin{bmatrix}F_1 & \cdots & F_K\end{bmatrix},
\quad
G = \begin{bmatrix}G_1 \\ \vdots \\ G_K\end{bmatrix},
\end{align*}
so that $F_i$ and $G_i$ are square matrices with shape $\frac{n}{K} \times \frac{n}{K}$ and
\begin{align*}
F \cdot G = \sum_{i=1}^{K} F_i \cdot G_i.
\end{align*}
The right hand side can be evaluated by $K$ matrix multiplications requiring $\mathcal{O}((\frac{n}{K})^\omega)$ (with $2 \le \omega$) field operations and $(K-1)(\frac{n}{K})^2$ additions of field elements. So the total time complexity is
\begin{align*}
\mathcal{O}(((K-1)(\tfrac{n}{K})^2 + K(\tfrac{n}{K})^\omega)\mathcal{Z}_q)
=\mathcal{O}(\tfrac{n^\omega}{K^{\omega-1}} \mathcal{Z}_q).
\end{align*}
\end{proof}

\begin{proposition} \label{pro: GASP time example}
Assume the setting of Proposition \ref{pro: server comp time} and that $K=L=n^\varepsilon$. Then, the time complexity for each operation in GASP is given in Table \ref{tab:poly optimized}.
\end{proposition}

\begin{proof}
The proof follows from substituting the values in the hypothesis and Proposition \ref{pro: server comp time} into Theorem \ref{teo: comp time gasp}.
\end{proof}

% In this section, we apply GASP to square matrices ($r=s=t=n$) and further assume $T=$ constant and $K=L=n^{\varepsilon}$ with $0 \le \varepsilon \le 1$ which implies $N \in \Theta(K^2)$, since for any polynomial code, we have the trivial bounds $KL \le N \le (K+T)(L+T)$ by the degree table. We also use a matrix multiplication algorithm as in Remark~\ref{rem:faster_multiplication_in_server} in Table~\ref{tab:poly optimized}, so that the asymptotic computational cost per server is $\frac{n^\omega}{K^{\omega-1}} = n^{\omega-\varepsilon(\omega-1)}$.

\begin{table}[h]
    \centering
    \begin{tabular}{ll}
    \toprule
    Operation & Time Complexity \\
    \midrule
    Query Encoding &   $\mathcal{O}(n^{2+2\varepsilon}\mathcal{Z}_q)$ \\
    User Decoding & $\mathcal{O}(n^{2+2\varepsilon}\mathcal{Z}_q)$ \\
    Each Server &   $\mathcal{O}(n^{\omega-\varepsilon(\omega-1)}\mathcal{Z}_q)$ \\
    \midrule
    Upload & $\mathcal{O}(n^{2+\varepsilon} \mathcal{C}_q)$ \\
    Download & $\mathcal{O}(n^2 \mathcal{C}_q)$ \\
    \bottomrule
    \end{tabular}
        \caption{Time complexity for the setting in Proposition \ref{pro: GASP time example}.}
            \label{tab:poly optimized}
\end{table}

We will now deal with the field size. Indeed, to use GASP codes we need the field size to satisfy certain bounds. Thus, by making $n$ grow, it will also be necessary to make the field size $q$ to grow. 

\begin{proposition}
Assume the setting in Proposition \ref{pro: GASP time example} and that ${\mathcal{Z}_q = \mathcal{O}(\log(q)^\gamma)}$. Then $\mathcal{Z}_q = \mathcal{O}(\log(n)^\gamma)$.
\end{proposition}

\begin{proof}
The proof for GASP codes in~\cite[Lemma~2]{9004505}, shows an argument for the evaluation points of $f$ and $g$ to exist if $q > \left(2\binom{N}{T}+1\right) \cdot J$.

Moreover, $J=\sum_{j \in \mathcal{J}} j$ where $\mathcal{J}$ is the set of exponents in $h(x)=\sum_{j \in \mathcal{J}} h_j x^j$ with $\#\mathcal{J} = N$. Since we use GASP codes, all entries in the degree table are between zero and \[W=2KL+(T-1)(K+1),\] so that, \[ J \le \sum_{j=0}^{W} j = \frac{W(W+1)}{2} = \mathcal{O}(W^2) = \mathcal{O}(K^2(L+T)^2) .\]

In particular, a field size larger than $\left(2\binom{N}{T}+1\right) \cdot \frac{W(W+1)}{2}$ is sufficient.

An application of the Stirling approximation ${n! \sim \sqrt{2 \pi n} (n/e)^n}$ yields
\begin{align*}
&
\binom{N}{T}
\le
\frac{N!}{T!(|N-T|)!}
=
\frac{\mathcal{O}(N!)}{\Omega(T!)\Omega((|N-T|)!)}
\\=&
\frac{\mathcal{O}(\sqrt{2 \pi N} (N/e)^N)}{\Omega(\sqrt{2 \pi T} (T/e)^T)\Omega(\sqrt{2 \pi |N-T|} (|N-T|/e)^{|N-T|})}
\\=&
\mathcal{O}\left(\frac{N^{N+\frac{1}{2}}}{\sqrt{2\pi} T^{T+\frac{1}{2}} |N-T|^{|N-T|+\frac{1}{2}}}e^{N-T-|N-T|}\right)
\\=&
\mathcal{O}\left(\frac{N^{N+\frac{1}{2}}}{\sqrt{2\pi} T^{T+\frac{1}{2}} |N-T|^{|N-T|+\frac{1}{2}}}\right),
\end{align*}
as $N-T-|N-T| \le 0$, so that
\begin{align*}
&\left(2\binom{N}{T}+1\right) \cdot \frac{W(W+1)}{2}
= \mathcal{O}\left(\binom{N}{T}W^2\right)
\\&=\mathcal{O}\left(\frac{N^{N+\frac{1}{2}}K^2(L+T)^2}{T^{T+\frac{1}{2}} |N-T|^{|N-T|+\frac{1}{2}}}\right)
\end{align*}
is a lower bound on the sufficient field size.

Using the same exemplary parameters as in Section~\ref{sec:param}, i.e., $T=$ constant, $K=L=n^{\varepsilon}$, and $N \in \mathcal{O}(K^2)$, noting that $T \le N$, this simplifies to $\mathcal{O}(n^{2\varepsilon(T+2)})$.
Using a field size in $\mathcal{O}(n^{2\varepsilon(T+2)})$ implies
\begin{align*}
\mathcal{Z}_q
= \mathcal{O}(\log(n^{2\varepsilon(T+2)})^\gamma)
= \mathcal{O}(\log(n)^\gamma).
\end{align*}
\end{proof}

We are now ready to calculate the total time complexity when implementing GASP codes.

\begin{theorem}
Assume the setting in Proposition \ref{pro: GASP time example}. Then, the total time complexity of GASP is $\mathcal{O}(n^{\max\{\varepsilon+\omega-\varepsilon\omega,2+2\varepsilon\} } \mathcal{Z}_q)$
\end{theorem}

\begin{proof}
We begin by noting that since $\mathcal{C}_q = \mathcal{O}(\log(q))$, it follows that $\mathcal{C}_q = \mathcal{O}(\mathcal{Z}_q)$.

% Further, we assume that $\mathcal{C}_q \in \mathcal{O}(\mathcal{Z}_q)$, which is reasonable, since if the communication and the computation cost of the prime field is considered as a constant, then $\mathcal{C}_q \in \mathcal{O}(\log(q))$ but $\mathcal{Z}_q \in \Omega(\log(q)^2)$ as, depending on the field representation, either addition or multiplication is in $\Theta(\log(q)^2)$.

% If, on the one hand, all servers to process their work sequential, then the overall time complexity is
% \begin{align*}
% &\mathcal{O}((K^2 n^2 + K^2 n^2 + N\tfrac{n^\omega}{K^{\omega-1}})\mathcal{Z}_q + (K n^2 + n^2)\mathcal{C}_q)
% \\&=
% \mathcal{O}(n^{\max\{3\varepsilon+\omega-\varepsilon\omega,2+2\varepsilon\} } \mathcal{Z}_q)
% .
% \end{align*}

Since all servers perform their computations in parallel, the total time complexity, $\mathcal{T}$, is the sum of the time complexities in Table \ref{tab:poly optimized},

\begin{align*}
\mathcal{T} &= \mathcal{O}((K^2 n^2 + K^2 n^2 + \tfrac{n^\omega}{K^{\omega-1}})\mathcal{Z}_q + (K n^2 + n^2)\mathcal{C}_q)
\\&=
\mathcal{O}(n^{\max\{\varepsilon+\omega-\varepsilon\omega,2+2\varepsilon\} } \mathcal{Z}_q)
.
\end{align*}
\end{proof}

The parameter $\varepsilon$ controls the trade-off between computational costs at the client ($\mathcal{O}(n^{2+2\varepsilon}\mathcal{Z}_q)$) versus computational costs at each of the servers ($\mathcal{O}(n^{\varepsilon+\omega-\varepsilon\omega}\mathcal{Z}_q)$). This trade-off, shown in Figure \ref{fig:Fa}, is linear in the exponents. By choosing $\varepsilon$ carefully we can bound the total time complexity, as shown in Figure \ref{fig:Fb}.

\begin{corollary} \label{cor: gasp}
Assume the setting in Proposition \ref{pro: GASP time example}. The minimum total time complexity for GASP is $\mathcal{O}(n^{4-\frac{6}{\omega+1}}\mathcal{Z}_q)$ for $\varepsilon=\frac{\omega-2}{\omega+1}$.
\end{corollary}

Thus, by using GASP codes, the user can perform the matrix multiplication in time $\mathcal{O}(n^{4-\frac{6}{\omega+1}}\mathcal{Z}_q)$ as opposed to the $\mathcal{O}(n^\omega)$ time it would take to do locally. Note here that since ${\mathcal{Z}_q = \mathcal{O}({\log(n)}^\gamma)}$, this is always an improvement. Also, if the user uses $\mathbb{F}_q$ as the base field, i.e. for very large fields, then $\mathcal{Z}_q$ can be taken to be constant.

Finally, we note that the analysis trivially holds for polynomial codes that outperform GASP codes since all results were proven via upper bounds. More so, since our analysis is done using asymptotic notation, the improvements would have to be by more than just constants to obtain better results.

\bibliography{main.bib}

\begin{thebibliography}{10}
\providecommand{\url}[1]{#1}
\csname url@samestyle\endcsname
\providecommand{\newblock}{\relax}
\providecommand{\bibinfo}[2]{#2}
\providecommand{\BIBentrySTDinterwordspacing}{\spaceskip=0pt\relax}
\providecommand{\BIBentryALTinterwordstretchfactor}{4}
\providecommand{\BIBentryALTinterwordspacing}{\spaceskip=\fontdimen2\font plus
\BIBentryALTinterwordstretchfactor\fontdimen3\font minus
  \fontdimen4\font\relax}
\providecommand{\BIBforeignlanguage}[2]{{%
\expandafter\ifx\csname l@#1\endcsname\relax
\typeout{** WARNING: IEEEtran.bst: No hyphenation pattern has been}%
\typeout{** loaded for the language `#1'. Using the pattern for}%
\typeout{** the default language instead.}%
\else
\language=\csname l@#1\endcsname
\fi
#2}}
\providecommand{\BIBdecl}{\relax}
\BIBdecl

\bibitem{ravi2018mmult}
W.-T. Chang and R.~Tandon, ``On the capacity of secure distributed matrix
  multiplication,'' in \emph{2018 IEEE Global Communications Conference
  (GLOBECOM)}.\hskip 1em plus 0.5em minus 0.4em\relax IEEE, 2018, pp. 1--6.

\bibitem{Kakar2019OnTC}
J.~Kakar, S.~Ebadifar, and A.~Sezgin, ``On the capacity and
  straggler-robustness of distributed secure matrix multiplication,''
  \emph{IEEE Access}, vol.~7, pp. 45\,783--45\,799, 2019.

\bibitem{d2019gasp}
R.~G.~L. D’Oliveira, S.~El~Rouayheb, and D.~Karpuk, ``Gasp codes for secure
  distributed matrix multiplication,'' in \emph{2019 IEEE International
  Symposium on Information Theory (ISIT)}.\hskip 1em plus 0.5em minus
  0.4em\relax IEEE, 2019, pp. 1107--1111.

\bibitem{DOliveira2019DegreeTF}
R.~G.~L. D'Oliveira, S.~{El Rouayheb}, D.~Heinlein, and D.~Karpuk, ``Degree
  tables for secure distributed matrix multiplication,'' in \emph{2019 IEEE
  Information Theory Workshop (ITW)}, 2019.

\bibitem{Aliasgari2019DistributedAP}
M.~Aliasgari, O.~Simeone, and J.~Kliewer, ``Distributed and private coded
  matrix computation with flexible communication load,'' \emph{2019 IEEE
  International Symposium on Information Theory (ISIT)}, pp. 1092--1096, 2019.

\bibitem{Kakar2019UplinkDownlinkTI}
J.~Kakar, A.~Khristoforov, S.~Ebadifar, and A.~Sezgin, ``Uplink-downlink
  tradeoff in secure distributed matrix multiplication,'' \emph{ArXiv}, vol.
  abs/1910.13849, 2019.

\bibitem{Yu2020EntangledPC}
Q.~Yu and A.~S. Avestimehr, ``Entangled polynomial codes for secure, private,
  and batch distributed matrix multiplication: Breaking the "cubic" barrier,''
  \emph{ArXiv}, vol. abs/2001.05101, 2020.

\bibitem{Chang2019OnTU}
W.-T. Chang and R.~Tandon, ``On the upload versus download cost for secure and
  private matrix multiplication,'' \emph{ArXiv}, vol. abs/1906.10684, 2019.

\bibitem{9004505}
R.~G.~L. {D’Oliveira}, S.~{El Rouayheb}, and D.~{Karpuk}, ``Gasp codes for
  secure distributed matrix multiplication,'' \emph{IEEE Transactions on
  Information Theory}, pp. 1--1, 2020.

\bibitem{jia2019capacity}
Z.~Jia and S.~A. Jafar, ``On the capacity of secure distributed matrix
  multiplication,'' 2019.

\bibitem{polycodes1}
Q.~Yu, M.~Maddah-Ali, and A.~S. Avestimehr, ``Polynomial codes: an optimal
  design for high-dimensional coded matrix multiplication,'' in \emph{Advances
  in Neural Information Processing Systems}, 2017, pp. 4403--4413.

\bibitem{polycodes2}
Q.~Yu, M.~A. Maddah-Ali, and A.~S. Avestimehr, ``Straggler mitigation in
  distributed matrix multiplication: Fundamental limits and optimal coding,''
  in \emph{2018 IEEE International Symposium on Information Theory
  (ISIT)}.\hskip 1em plus 0.5em minus 0.4em\relax IEEE, 2018, pp. 2022--2026.

\bibitem{pulkit}
S.~Dutta, M.~Fahim, F.~Haddadpour, H.~Jeong, V.~Cadambe, and P.~Grover, ``On
  the optimal recovery threshold of coded matrix multiplication,'' \emph{IEEE
  Transactions on Information Theory}, 2019.

\bibitem{pulkit2}
U.~Sheth, S.~Dutta, M.~Chaudhari, H.~Jeong, Y.~Yang, J.~Kohonen, T.~Roos, and
  P.~Grover, ``An application of storage-optimal matdot codes for coded matrix
  multiplication: Fast k-nearest neighbors estimation,'' in \emph{2018 IEEE
  International Conference on Big Data (Big Data)}.\hskip 1em plus 0.5em minus
  0.4em\relax IEEE, 2018, pp. 1113--1120.

\bibitem{fundamental}
S.~Li, M.~A. Maddah-Ali, Q.~Yu, and A.~S. Avestimehr, ``A fundamental tradeoff
  between computation and communication in distributed computing,'' \emph{IEEE
  Transactions on Information Theory}, vol.~64, no.~1, pp. 109--128, 2017.

\bibitem{howell2008asymptotic}
R.~R. Howell, ``On asymptotic notation with multiple variables,'' \emph{Tech.
  Rep.}, 2008.

\bibitem{gashkov2013complexity}
S.~B. Gashkov and I.~S. Sergeev, ``Complexity of computation in finite
  fields,'' \emph{Journal of Mathematical Sciences}, vol. 191, no.~5, pp.
  661--685, 2013.

\bibitem{katz1996handbook}
J.~Katz, A.~J. Menezes, P.~C. Van~Oorschot, and S.~A. Vanstone, \emph{Handbook
  of applied cryptography}.\hskip 1em plus 0.5em minus 0.4em\relax CRC press,
  1996.

\bibitem{brent2010modern}
R.~P. Brent and P.~Zimmermann, \emph{Modern computer arithmetic}.\hskip 1em
  plus 0.5em minus 0.4em\relax Cambridge University Press, 2010, vol.~18.

\bibitem{strassen1969gaussian}
V.~Strassen, ``Gaussian elimination is not optimal,'' \emph{Numerische
  {M}athematik}, vol.~13, no.~4, pp. 354--356, 1969.

\bibitem{le2014powers}
F.~Le~Gall, ``Powers of tensors and fast matrix multiplication,'' in
  \emph{Proceedings of the 39th international symposium on symbolic and
  algebraic computation}.\hskip 1em plus 0.5em minus 0.4em\relax ACM, 2014, pp.
  296--303.

\bibitem{coppersmith1990matrix}
D.~Coppersmith and S.~Winograd, ``Matrix multiplication via arithmetic
  progressions,'' \emph{Journal of symbolic computation}, vol.~9, no.~3, pp.
  251--280, 1990.

\bibitem{chor1995private}
B.~Chor, O.~Goldreich, E.~Kushilevitz, and M.~Sudan, ``Private information
  retrieval,'' in \emph{Proceedings of IEEE 36th Annual Foundations of Computer
  Science}.\hskip 1em plus 0.5em minus 0.4em\relax IEEE, 1995, pp. 41--50.

\bibitem{8906051}
R.~G.~L. {D’Oliveira} and S.~{El Rouayheb}, ``{One-Shot PIR: Refinement and
  Lifting},'' \emph{IEEE Transactions on Information Theory}, pp. 1--1, 2019.

\bibitem{sun2017capacity}
H.~Sun and S.~A. Jafar, ``The capacity of robust private information retrieval
  with colluding databases,'' \emph{IEEE Transactions on Information Theory},
  vol.~64, no.~4, pp. 2361--2370, 2017.

\end{thebibliography}
\bibliographystyle{IEEEtran}

\end{document}